\documentclass[11pt]{article}
\usepackage[utf8]{inputenc}

\usepackage[T1]{fontenc}

\usepackage{fullpage}
\usepackage{amsmath,amsfonts,amsthm,xspace,hyperref,graphicx,relsize,bm}
\usepackage{mathpazo}
\usepackage[margin=1in]{geometry}
\usepackage{float}
\usepackage{endnotes}
\usepackage{enumitem}
\usepackage{xcolor}
\usepackage{cleveref}
\usepackage{stmaryrd}
\usepackage{dsfont}
\usepackage{thm-restate}

\usepackage{url}
\usepackage{qcircuit}
\usepackage{upgreek}

\newcommand{\suppress}[1]{}
\newcommand{\comment}[1]{}

\definecolor{DarkBlue}{rgb}{0,0,0.8}  
\definecolor{DarkOrange}{rgb}{0.8,0.4,0}  
\def\mylinkcolor{DarkBlue}

\hypersetup{
        colorlinks=true, urlcolor=\mylinkcolor,
  linkcolor=\mylinkcolor, citecolor=\mylinkcolor}

\usepackage{mathtools}

\newtheorem{theorem}{Theorem}[section]
\newtheorem{proposition}[theorem]{Proposition}
\newtheorem{lemma}[theorem]{Lemma}

\newtheorem{definition}[theorem]{Definition}
\theoremstyle{definition}

\newtheorem{example}[theorem]{Example}
\newtheorem{remark}[theorem]{Remark}

\newtheorem{conjecture}[theorem]{Conjecture}

\numberwithin{equation}{section}

\newcommand{\complex}{{\mathbb C}}
\newcommand{\reals}{{\mathbb R}}

\newcommand{\tensor}{\otimes}

\newcommand{\adjoint}{\dagger}
\newcommand{\eqdef}{\coloneqq}

\newcommand{\ket}[1]{| #1 \rangle}
\newcommand{\bra}[1]{ \langle #1 |}
\newcommand{\ketbra}[2]{| #1 \rangle\!\langle #2 |}
\newcommand{\braket}[2]{\langle #1 | #2 \rangle }
\newcommand{\density}[1]{\ketbra{#1}{#1}}
\newcommand{\ketbraq}[1]{\ketbra{#1}{#1}}

\newcommand{\norm}[1]{\left\| #1 \right\|}

\newcommand{\size}[1]{\left| #1 \right|}
\newcommand{\set}[1]{\left\{ #1 \right\}}

\newcommand{\complexi}{{\mathrm{i}}}
\newcommand{\ii}{{\mathrm{i}}}
\newcommand{\bi}{{\mathbf{i}}}
\newcommand{\bj}{{\mathbf{j}}}
\newcommand{\bk}{{\mathbf{k}}}
\newcommand{\abs}[1]{\left| #1 \right|}
\newcommand{\id}{{\mathbb 1}}
\newcommand{\I}{{\mathbb 1}}
\newcommand{\bC}{\mathbb{C}}
\newcommand{\bH}{\mathbb{H}}
\newcommand{\tran}[0]{^\textnormal{\tiny{T}}}
\newcommand{\hc}[0]{^{\dagger}}

\DeclareMathOperator{\tr}{tr}

\DeclareMathOperator{\trace}{Tr}

\newcommand{\PauliX}{X}
\newcommand{\PauliY}{Y}
\newcommand{\PauliZ}{Z}
\newcommand{\mes}{\upphi} 

\newcommand{\bq}{{\bm q}}
\newcommand{\ba}{{\bm a}}
\newcommand{\bb}{{\bm b}}
\newcommand{\bx}{{\bm x}}
\newcommand{\by}{{\bm y}}
\newcommand{\bz}{{\bm z}}
\newcommand{\bw}{{\bm w}}
\newcommand{\cH}{\mathcal{H}}

\newcommand{\tU}{\widetilde{U}}
\newcommand{\sQ}{{\mathsf{Q}}}
\newcommand{\sP}{{\mathsf{P}}}

\title{\textbf{Mutually Unbiased Measurements, Hadamard Matrices, and Superdense Coding}}

\author{
M\'{a}t\'{e} Farkas~\thanks{ICFO-Institut de Ciencies Fotoniques, The Barcelona Institute of Science and Technology, 08860 Castelldefels, Spain. Email: \texttt{mate.farkas@icfo.eu}~.}
\and 
J\k{e}drzej Kaniewski~\thanks{Faculty of Physics, University of Warsaw, Pasteura 5, 02-093 Warsaw, Poland. Email: \texttt{jkaniewski@fuw.edu.pl}~.}
\and 
Ashwin Nayak~\thanks{Department of Combinatorics and Optimization,
and Institute for Quantum Computing,
University of Waterloo, 200 University Ave.\ W., Waterloo, ON,
N2L~3G1, Canada. Email: \texttt{ashwin.nayak@uwaterloo.ca}~.}
}

\date{\bf April 7, 2023}

\begin{document}

\maketitle

\begin{abstract}
Mutually unbiased bases (MUBs) are highly symmetric bases on complex Hilbert spaces, and the corresponding rank-1 projective measurements are ubiquitous in quantum information theory. In this work, we study a recently introduced generalisation of MUBs called \textit{mutually unbiased measurements} (MUMs). These measurements inherit the essential property of complementarity from MUBs, but the Hilbert space dimension is no longer required to match the number of outcomes. This operational complementarity property renders MUMs highly useful for device-independent quantum information processing. 

It has been shown that MUMs are strictly more general than MUBs. In this work we provide a complete proof of the characterisation of MUMs that are direct sums of MUBs. We then construct new examples of MUMs that are not direct sums of MUBs. A crucial technical tool for this construction is a correspondence with quaternionic Hadamard matrices, which allows us to map known examples of such matrices to MUMs that are not direct sums of MUBs. Furthermore, we show that---in stark contrast with MUBs---the number of MUMs for a fixed outcome number is unbounded. 

Next, we focus on the use of MUMs in quantum communication. We demonstrate how any pair of MUMs with~$d$ outcomes defines a~$d$-dimensional superdense coding protocol. Using MUMs that are not direct sums of MUBs, we disprove a recent conjecture due to Nayak and Yuen on the rigidity of superdense coding, for infinitely many dimensions. The superdense coding protocols arising in the refutation reveal how shared entanglement may be used in a manner heretofore unknown.
\end{abstract}

\section{Introduction}
\label{sec-intro}

A pair of bases for a~$d$-dimensional Hilbert space is said to be \emph{mutually unbiased\/} if the squared length of the projection of any basis element from the first onto any basis element of the second is exactly~$1/d$. Mutually unbiased bases (MUBs)~\cite{DEBZ10} are well-studied objects in quantum information theory, and they are optimal in several tasks, including some in quantum cryptography. Thus, they correspond to a fundamentally important class of measurements.

Mutually unbiased \emph{measurements\/} (MUMs) were recently introduced by Tavakoli, Farkas, Rosset, Bancal, and Kaniewski~\cite{TFR+21} as a generalization of mutually unbiased bases. The optimality of MUBs in information processing tasks can be traced back to their \emph{complementarity relations\/}. MUMs were introduced in the context of Bell inequalities, and it was shown that MUMs have the same complementarity relations as MUBs. Furthermore, MUMs admit the same entropic uncertainty relations and the same measurement incompatibility robustness as MUBs. However, MUMs are defined in a ``device-independent'' manner, i.e., without reference to the dimension of the Hilbert space on which they act. This makes them particularly useful in the context of device-independent cryptography.

Note that an equivalent definition of MUMs based on complementarity has been introduced earlier in Ref.~\cite{TSWR18}. The authors focus on continuous-variable systems and analyse the properties of MUMs in this setting in subsequent works \cite{PWTR18,RW21,SRTW22}. We also note here that an alternative, inequivalent definition of MUMs was introduced in Ref.~\cite{KG14}. The authors define MUMs through uniform overlaps of the measurements operators, depending on an ``efficiency parameter''. Their definition reduces to MUBs when the value of the efficiency parameter is 1, but the authors show that many MUMs can be constructed with lower efficiency parameter. For a detailed account on the difference between the MUM definition in Ref.~\cite{KG14} and in Ref.~\cite{TFR+21} (which is the one used in the current paper), see Remark~\ref{remark:KG14}.

Despite all their similarities, it has been shown that MUMs are strictly more general than MUBs: there exist pairs of MUMs that are not ``direct sums'' of MUBs, and even pairs of MUMs that cannot be mapped to a pair of MUBs by means of a completely positive unital map \cite{TFR+21}.

In this work, we expand on the study of MUMs. We first provide an improved and complete proof of the characterization of pairs of MUMs that are direct sums of MUBs (\Cref{prop-direct-sum-mub}), originally proposed in Ref.~\cite{TFR+21}. We draw a connection between pairs of MUMs and Hadamard matrices of unitary operators, and further between the latter and quaternionic Hadamard matrices (\Cref{thm-non-directsum}). This enables us to construct novel examples of MUM pairs that are not direct sums of MUBs for new sets of outcome numbers (\Cref{sec-new-mums}). Using the same correspondence, we also present an infinite family of such MUMs with four outcomes. Earlier, only two examples of such MUM pairs were known, for four or five outcomes. Finally, we prove that---in stark contrast with MUBs---there exist arbitrarily large collections of pairwise MUMs for every outcome number (\Cref{thm:no_MUMs}). 

Along with the literature on quaternionic Hadamard matrices, the recipe we provide for the construction of MUM pairs suggests that MUM pairs that are not direct sums of MUBs are more common than one might be led to believe. They are likely not particular to special sets of outcome numbers, and for each outcome number greater than three (and underlying dimension), it is likely that there exist an infinite number of non-equivalent such pairs. For outcome numbers two and three, it has been shown that any pair of MUMs is a direct sum of MUBs \cite{TFR+21}.

In the second part of this work, we turn our attention to superdense coding protocols. The original two-party communication protocol due to Bennett and Wiesner~\cite{BW92-sdc} uses a shared EPR pair and transmits two bits of classical information by sending only \emph{one\/} qubit. This generalizes to~$d$-dimensional protocols for arbitrary~$d \ge 2$ (\Cref{def:sdc-protocol}). These protocols use a maximally entangled state of local dimension~$d$ and transmit an arbitrary message out of~$d^2$ possibilities by sending a~$d$-dimensional quantum state. Superdense coding further generalizes to the transmission of quantum states, and plays an important role in Quantum Shannon Theory (see, e.g.,~\cite[Chapter~6]{W13-qit}).

Given one~$d$-dimensional superdense coding protocol, we may construct other equivalent protocols. For example, the two parties may use a maximally entangled state which is obtained by applying arbitrary unitary operators to the two halves of the shared state. We may also ``mix'' two or more different protocols to construct a seemingly more complicated protocol. Namely, we may use an entangled state of larger dimension to generate a common random string shared by the two parties, and use the shared randomness to select one of several superdense coding protocols to transmit the classical message. Finally, the sender may apply an arbitrary unitary operator, possibly depending on the classical message, to the part of the entangled state that is retained by her. Nayak and Yuen~\cite{NY20-rigidity} recently showed that any two-dimensional superdense coding protocol may be obtained by performing these transformations on the original Bennett-Wiesner protocol. In other words, they showed that the Bennett-Wiesner protocol is \emph{rigid\/}. They further conjectured that a similar rigidity property holds for higher dimensional protocols ($d > 2$), up to the choice of the basic protocol (which uses a maximally entangled state of local dimension~$d$); see \Cref{conj:high-dim-rigidity}.

We disprove the rigidity conjecture due to Nayak and Yuen for infinitely many dimensions~$d \ge 4$. We do this by first showing that any pair of MUMs with~$d$ outcomes defines a~$d$-dimensional superdense coding protocol (\Cref{thm-mum-protocol}). We then prove that the rigidity conjecture implies that the MUMs used in the protocol are a direct sum of MUBs (\Cref{thm-rigidity-direct-sum}). The examples of MUMs that we construct (\Cref{sec-mum-examples,sec-new-mums}) now give us counterexamples for infinite sequences of dimensions~$d$.

The precise power of shared entanglement in quantum communication complexity has been a long-standing open problem (see, e.g., Ref.~\cite{Gavinsky08-entanglement}). Shared entanglement can be used to reduce communication complexity by a factor of two, using superdense coding. It may also be used to generate shared randomness, which in turn can be used to reduce communication complexity of boolean functions on~$n$-bits by an additive~$\log n$ term. It is not known whether it leads to a reduction in the communication complexity of computing functions or relations beyond these two phenomena. The superdense coding protocols arising in the counterexamples above show that shared entanglement may be used in a manner heretofore unknown. In particular, it hints at the possibility of new, counter-intuitive ways in which it might aid communication.

\paragraph{Acknowledgements.}
We are grateful to Henry Yuen for several helpful discussions, and to Curtis Bright for explaining the construction of the perfect sequences in Ref.~\cite{BKG20-perfect-sequences}. This  project  has  received  funding  from  the  European  Union’s  Horizon~2020  research and innovation programme under the Marie Sk\l{}odowska-Curie grant agreement No.~754510. M.F.~acknowledges funding from the Government of Spain (FIS2020-TRANQI, Severo Ochoa CEX2019-000910-S), Fundaci\'o Cellex, Fundaci\'o Mir-Puig and Generalitat de Catalunya (CERCA, AGAUR SGR 1381). J.K.~acknowledges support from the National Science Centre, Poland under the SONATA project ``Fundamental aspects of the quantum set of correlations'' (grant no.~2019/35/D/ST2/02014). A.N.'s research is supported in part by a Discovery Grant from NSERC Canada. 

\section{Mutually unbiased measurements}
\label{sec-mum}

We define the notion of mutually unbiased measurements and describe several properties of relevance to us in \Cref{sec-mum-properties,sec-mum-mub}. We present some explicit examples of MUMs in \Cref{sec-mum-examples}. We then describe a recipe for constructing MUMs based on generalized Hadamard matrices in \Cref{sec-mum-hadamard}. Finally, we describe explicit MUMs based on this construction in \Cref{sec-new-mums}.

\subsection{Definition and basic properties}
\label{sec-mum-properties}

For any positive integer~$n$, let~$[n]$ denote the set~$\set{1, 2, \dotsc, n}$. For any integer~$d \ge 2$, let~$\omega_d \coloneqq \exp \left( \tfrac{2 \pi \complexi}{d} \right)$  be a primitive~$d$th root of unity. Consider the standard basis~$ \set{ \ket{j} : j \in [d]}$ and the Fourier basis~$ \set{ \ket{\chi_j} : j \in [d] } $ for~$\complex^d$, where~$\ket{\chi_j} \eqdef \tfrac{1}{\sqrt{d}} \sum_{i = 1}^d \omega_d^{ij} \ket{i}$. We have~$\size{ \braket{j}{\chi_k}}^2 = \tfrac{1}{d}$ for any~$j,k \in [d]$, and the bases are said to be \emph{mutually unbiased\/}. Formally, we define MUBs in terms of the corresponding measurements.
\begin{definition}
	Two $d$-outcome measurements $\{P_a\}_{a=1}^d$ and $\{Q_b\}_{b=1}^d$ acting on $\bC^d$ correspond to \emph{mutually unbiased bases} (MUBs) if the measurement operators are rank-one orthogonal projections, i.e., $P_{a} = \ketbraq{u_{a}}$ and $Q_{b} = \ketbraq{v_{b}}$ for some pure states~$\ket{u_a}, \ket{v_b} \in \complex^d$, and
	\begin{equation}
	    \abs{ \braket{u_a}{v_b} }^2 \quad = \quad \frac1d \qquad \forall a,b \in [d] \enspace.
	\end{equation}
\end{definition}
Since~$\sum_a P_a = \id$, we see that~$\set{ \ket{u_a} }$ and~$\set{ \ket{v_b} }$ are orthonormal bases for~$\complex^d$. 

We consider a generalization of MUBs, called MUMs \cite{TFR+21}. The definition of MUMs is ``device-independent'' in the sense that it does not refer to the dimension of the Hilbert space on which they act.
\begin{definition}\label{def:MUM}
	Two $d$-outcome measurements $\{P_a\}_{a=1}^d$ and $\{Q_b\}_{b=1}^d$ acting on a Hilbert space $\cH$ are called \emph{mutually unbiased measurements} (MUMs) if
	\begin{equation}\label{eq:MUM_app}
		P_{a} \quad = \quad dP_{a} Q_{b} P_{a} \qquad \text{and}  \qquad  Q_{b} \quad = \quad dQ_{b} P_{a} Q_{b} \enspace,
	\end{equation}
	for all $a$ and $b$ in~$[d]$.
\end{definition}
We may verify readily that every pair of MUBs is also a pair of MUMs. 

\begin{remark}\label{remark:KG14}
We note here that a different notion of MUMs was introduced by Kalev and Gour in Ref.~\cite{KG14}. Their definition is different from Definition~\ref{def:MUM} in a few ways: they fix the Hilbert space dimension as well as the number of measurement outcomes (albeit these can be different). They then require the trace of each operator to be 1, and all of the traces of the form $\tr(P_a Q_b)$ to be uniform, dependent only on an ``efficiency parameter'' (at the same time, the traces $\tr(P_a P_b)$ and $\tr(Q_a Q_b)$ depend only on whether $a=b$, and the efficiency parameter). The measurements of Kalev and Gour are projective only if the efficiency parameter is 1, in which case the measurements are also MUBs. As one can verify that MUMs in Definition~\ref{def:MUM} are always projective, our MUM definition is different from that of Kalev and Gour whenever the measurements do not correspond to MUBs. As such, in the following we will use the term MUM exclusively in the sense of Definition~\ref{def:MUM}.
\end{remark}

Given~$d$-outcome MUMs, we can construct~$d'$-outcome MUMs for any multiple of~$d$ by taking a tensor product with the elements of a suitable MUB.
\begin{lemma}
\label{lem-multiple-dim-mum}
Let~$\{P_a\}_{a=1}^d$ and $\{Q_b\}_{b=1}^d$ be a pair of $d$-outcome MUMs on Hilbert space~$\cH$. Then for any integer~$\ell > 1$, there is a pair of~$d'$-outcome MUMs on Hilbert space~$\cH \tensor \complex^\ell $, where~$d' \eqdef \ell d$.
\end{lemma}
\begin{proof}
Define orthogonal projection operators
\begin{align*}
P'_{a,i} & \quad \eqdef \quad P_a \tensor \density{i} \enspace, \qquad \text{and} \\
Q'_{b,j} & \quad \eqdef \quad Q_b \tensor \density{\chi_j} \enspace,
\end{align*}
for~$a,b \in [d]$ and~$i,j \in [\ell ]$, where~$(\ket{\chi_j} : j \in [\ell ])$ is the Fourier basis for~$\complex^\ell $. We may verify that~$\set{P'_{a,i}}$ and $\set{Q'_{b,j}}$ are a pair of~$d'$-outcome MUMs.
\end{proof}

Any pair of (finite dimensional) MUMs can be written as a pair of $d$ orthogonal projections on the tensor product space $\bC^n \otimes \bC^d$ for some $n\in\mathbb{N}$ that satisfy certain \emph{MUM conditions\/}~\cite{TFR+21} (see also Ref.~\cite{NPA12-sdp-relaxations}). I.e., up to a change of basis, the first measurement is simply given by
\begin{equation}\label{eq:Pa}
P_{a}  \quad = \quad  \I \otimes \ketbraq{a} \enspace,
\end{equation}
where $\{ \ket{a} \}_{a = 1}^d$ is the computational basis on $\bC^d$. In the same basis used to express~$P_a$, the second measurement can be described in full generality by
\begin{equation}\label{eq:Qb}
Q_{b}  \quad = \quad  \frac{1}{d} \sum_{j,k \in [d]} V^b_{jk} \otimes \ketbra{j}{k}
\end{equation}
for each~$b \in [d]$, where $V^b_{jk}$ are linear operators on $\bC^n$. The MUM conditions are then given by the following relations in terms of the operators~$V^b_{jk}$~\cite[Supplementary materials, Sec.~II.B.1]{TFR+21}:
\begin{align}
\begin{split}
\label{eq:Xbjk}
V^b_{jj}  & \quad = \quad  \I \qquad \forall b,j \\
(V^b_{jk})\hc  & \quad = \quad  V^b_{kj} \qquad \forall b,j,k \\
V^b_{jk}  & \quad = \quad  V^b_{jl} V^b_{lk} \qquad \forall b,j,k,l \\
\sum_b V^b_{jk}  & \quad = \quad  \updelta_{jk} d \, \I \qquad \forall j,k \enspace.
\end{split}
\end{align}
The last condition corresponds to the completeness relation $\sum_b Q_b = \I$. Notice that the MUM conditions imply that $Q_b$ are projections, and therefore the completeness relation is sufficient to ensure that $\set{Q_b}$ form a valid measurement. While superfluous, the orthogonality constraint $Q_b Q_{b'} = \updelta_{bb'} Q_b$ corresponds to
\begin{equation}\label{eq:orthogonality_X}
\sum_{l} V^b_{jl} V^{b'}_{lk}  \quad = \quad  \updelta_{bb'} d V^b_{jk} \enspace.
\end{equation}

The conditions in Eq.~\eqref{eq:Xbjk} make it possible to describe the pair of MUMs in terms of $d^2$ unitary operators
\begin{equation}
\label{eq-mum-u}
U^b_j \quad \eqdef \quad V^b_{j1}  \qquad \text{for } b,j \in [d] \enspace,
\end{equation}
since $V^b_{jk} = U^b_j ( U^b_k )\hc$ for all $b,j,k$. The MUM conditions in terms of $\set{ U^b_j }$ are given by
\begin{equation}
\label{eq:MUM_U}
\begin{split}
U^b_1  & \quad = \quad  \I \qquad \forall b \\
U^b_j ( U^b_j )\hc  & \quad = \quad  ( U^b_j )\hc U^b_j  \quad = \quad  \I \qquad \forall b,j \\
\sum_b U^b_j ( U^b_k )\hc  & \quad = \quad  \updelta_{jk} d \, \I \qquad \forall j,k \enspace.
\end{split}
\end{equation}
The orthogonality condition in terms of the unitary operators~$\set{ U^b_j }$ is given by
\begin{equation}\label{eq:orthogonality_U}
\sum_j (U^b_j)\hc U^{b'}_j  \quad = \quad  \updelta_{bb'} d \, \I \qquad \forall b, b' \in [d] \enspace.
\end{equation}
Since unitary operators do not commute in general, the order of multiplication of matrices in Eqs.~\eqref{eq:MUM_U} and~\eqref{eq:orthogonality_U} is important.
If in addition to the MUM conditions in Eq.~(\ref{eq:MUM_U}), the operators~$U^b_j$ satisfy
\begin{equation}
\label{eq-canonical-U}
U^1_j  \quad = \quad \id \qquad \forall j \in [d] \enspace,
\end{equation}
we say that the MUMs are in \emph{canonical form\/}. We may convert any pair of MUMs into canonical form by a simple transformation.
\begin{lemma}
\label{lem-canonical-form}
Suppose $\set{P_a}$ and $\set{Q_b}$ are a pair of~$d$-outcome MUMs defined by operators~$(U^b_j)$ satisfying the MUM conditions in Eq.~(\ref{eq:MUM_U}). Let~$U \eqdef \sum_{j = 1}^d (U^1_j)^\adjoint \tensor \density{j}$. Then~$\set{ U P_a U^\adjoint}$ and $\set{ U Q_b U^\adjoint}$ are a pair of~$d$-outcome MUMs in canonical form.
\end{lemma}
The transformation essentially replaces the operator~$U^b_j$ by~$\tU^b_j \eqdef (U^1_j)^\adjoint U^b_j$. We leave it to the reader to verify that the operators~$( \tU^b_j )$ satisfy the MUM conditions, and have~$\tU^1_j = \id$ for all~$j \in [d]$. A further generalization of this lemma is useful for the characterization of MUMs.

\begin{proposition}\label{prop:dephasing}
The conditions
\begin{equation}\label{eq:unitary_Hadamard}
\begin{split}
(U^b_j)(U^b_j)^\dagger \quad & \left. = \quad (U^b_j)^\dagger U^b_j \quad = \quad \id \qquad \forall b,j \right. \\
\sum_b U^b_j (U^b_k)^\dagger \quad & \left. = \quad \updelta_{jk} d \, \id \qquad \forall j,k \right. \\
\sum_j (U^b_j)^\dagger U^{b'}_j \quad  & \left. = \quad \updelta_{bb'} d \, \id \qquad \forall b,b' \right.
\end{split}
\end{equation}
are stable under the transformations
\begin{equation}
\begin{split}
U^b_j \quad & \left. \mapsto \quad U^b_j W^b \right. \\
U^b_j \quad & \left. \mapsto \quad W_j U^b_j \right.
\end{split}
\end{equation}
where $W^b$ and $W_j$ are arbitrary unitary matrices.
\end{proposition}

\subsection{Direct sum property}
\label{sec-mum-mub}

MUMs are strictly more general than MUBs. In particular, not all pairs of MUMs can be written as a \emph{direct sum\/} of MUBs.
\begin{definition}
	We say that two $d$-outcome measurements, $\{P_a\}_{a=1}^d$ and $\{Q_b\}_{b=1}^d$ on $\cH$, are a \emph{direct sum of mutually unbiased bases\/} if $\cH \cong \bigoplus_{j} \bC^{d}$ and $P_{a} = \bigoplus_{j} P_{a}^{j}$, $Q_{b} = \bigoplus_{j} Q_{b}^{j}$, where for every $j$ the pair $\{P^{j}_a\}_{a=1}^d$ and $\{Q^{j}_b\}_{b=1}^d$ are mutually unbiased bases acting on the~$j$th direct summand of~$\cH$.
\end{definition}

When expressed in the canonical form described in \Cref{sec-mum-properties}, it is straightforward to check if the pair of MUMs $\set{P_a}$ and $\set{Q_b}$ are a direct sum of MUBs due to the equivalence below, stated in Proposition~II.6 in the supplementary materials for Ref.~\cite{TFR+21}. The proof therein simply claimed the ``only if'' direction of this characterisation to be ``clear'' and only gave an argument for the ``if'' direction. We provide a complete proof, including that for the ``only if'' direction (i.e, the converse). In fact, the converse crucially relies on the operators~$(U^b_j)$ being in canonical form; it provably does not hold otherwise. This subtlety is fully reflected in the proof we provide. (As we show in \Cref{sec-mum-protocol}, the converse is the direction of relevance to the rigidity of superdense coding protocols.)

\begin{restatable}{proposition}{directsummub}
\label{prop-direct-sum-mub}
Suppose $\set{P_a}$ and $\set{Q_b}$ are a pair of~$d$-outcome MUMs in canonical form defined by operators~$(U^b_j : b,j \in [d])$ satisfying the MUM conditions in Eq.~(\ref{eq:MUM_U}) and the canonical form condition in Eq.~\eqref{eq-canonical-U}.
Then $\set{P_a}$ and $\set{Q_b}$ are a direct sum of MUBs if and only if
\begin{equation}\label{eq:U_comm}
[U^b_j \,, U^{b'}_{j'}] \quad = \quad 0 \qquad \forall b,b',j,j' \enspace.
\end{equation}
\end{restatable}
\begin{proof}
Suppose the MUMs~$\set{P_a}$ and~$\set{Q_b}$ are defined by the operators~$( U^b_j : b,j \in [d])$ in canonical form. I.e., we fix a choice of basis in which the operators~$P_a$ are expressed as
\begin{equation*}
P_a \quad = \quad \id \otimes \ketbraq{a} \enspace,
\end{equation*}
where the second Hilbert space factor is isomorphic to $\bC^d$ and $\set{\ket{a}}$ is an orthonormal basis for it. Furthermore, the $Q_b$ operators are expressed as
\begin{equation*}
Q_b \quad = \quad \frac{1}{d} \sum_{j,k \in [d]} U^b_j (U^b_k)\hc \otimes \ketbra{j}{k} \enspace,
\end{equation*}
where the operators~$U^b_j$ satisfy the conditions in Eq.~\eqref{eq:MUM_U} and Eq.~\eqref{eq-canonical-U}. Suppose, without loss of generality, that the Hilbert space on which the MUMs are defined is isomorphic to~$\cH \eqdef \complex^{d'} \tensor \complex^d$ in the basis used for the canonical form, for some positive integer~$d'$.

We start with the ``if'' direction of the characterization. Suppose the operators~$U^b_j$ all commute with each other. Then there is an orthonormal basis~$\set{ \ket{e_t} : t \in [d']}$ for~$\complex^{d'}$ in which the operators~$ U^b_j$ are simultaneously diagonal. I.e.,
\[
U^b_j \quad = \quad \sum_{t = 1}^{d'} \lambda^b_{jt} \; \density{e_t} \enspace,
\]
for some unit complex numbers~$\lambda^b_{jt}$, for all~$b,j \in [d]$. Consider the direct sum decomposition
\[
\cH \quad = \quad \bigoplus_{t = 1}^{d'} \density{e_t} \tensor \complex^d \enspace,
\]
and the orthogonal projection operators~$R_t \eqdef \density{e_t} \tensor \id$ into the~$t$-th direct summand, for~$t \in [d']$. Define~$P_a^t \eqdef R_t P_a R_t = \density{e_t} \tensor \density{a}$, and
\[
Q_b^t \quad \eqdef \quad R_t Q_b R_t \quad = \quad \density{e_t} \tensor \; \frac{1}{d} \sum_{jk} \lambda^b_{jt} \overline{\lambda^b_{kt}} \; \ketbra{j}{k} \enspace. 
\]
We have~$P_a = \bigoplus_{t \in [d']} P_a^t$ and~$Q_b = \bigoplus_{t \in [d']} Q_b^t$.
Moreover,~$P_a^t, Q_b^t$ are rank 1 projection operators and~$\trace(P_a^t Q_b^t) = 1/d$, for all~$t,a,b$. So the pair of MUMs~$\set{P_a}$ and~$\set{Q_b}$ are a direct sum of MUBs.

To prove the ``only if'' direction, suppose that~$\set{P_a}$ and~$\set{Q_b}$ are a direct sum of MUBs. Then~$\cH$ is a direct sum of~$d'$ subspaces~$\cH_t$ ($t \in [d']$), and each subspace~$\cH_t$ has dimension~$d$. Let~$S_t$ be the orthogonal projection operator onto~$\cH_t$. 

By the direct sum property, we have
\begin{align*}
P_a \quad & = \quad \sum_t S_t P_a S_t \enspace, \qquad \text{and} \\
Q_b \quad & = \quad \sum_t S_t Q_b S_t \enspace,
\end{align*}
for all~$a,b$. So~$S_t P_a = S_t P_a S_t = P_a S_t$ and similarly, the operators~$S_t$ and~$Q_b$ also commute. In particular, $S_t$ and~$P_a$ are simultaneously diagonalisable. Since the eigenvectors of~$P_a$ are all of the form~$\ket{v} \ket{a}$ for some~$\ket{v} \in \complex^{d'}$, we may express~$S_t$ as~$S_t = \sum_a S_{ta} \tensor \density{a}$, for some orthogonal projection operators~$S_{ta}$. 

Since the MUMs are in canonical form, we have~$U^1_j = \id$ for all~$j$. So~$Q_1 = \id \tensor (1/d) \sum_{j,k} \ketbra{j}{k}$. From the commutation of~$S_t$ and~$Q_1$ we have
\[
\frac{1}{d} \sum_{j,k} S_{tj} \tensor \ketbra{j}{k} 
    \quad = \quad S_t Q_1 
    \quad = \quad Q_1 S_t
    \quad = \quad \frac{1}{d} \sum_{j,k} S_{tk} \tensor \ketbra{j}{k} \enspace.
\]
So~$S_{tj} = S_{tk}$ for all~$j,k$. Let~$S_t' \eqdef S_{t1}$. Then~$S_t = S_t' \tensor \id$, and~$\sum_t S_t' = \id$. 

Since~$\set{S_t P_a S_t : a \in [d]}$ and~$\set{S_t Q_b S_t : b \in [d]}$ are MUBs for every~$t$, these projection operators are all rank~$1$. As~$S_t P_a S_t = S_t' \tensor \density{a}$, we infer that~$\set{S_t'}$ are also rank~$1$ and we have~$S_t' = \density{f_t}$ for some orthonormal basis~$\set{\ket{f_t}}$ for~$\complex^{d'}$. Since~$Q_b = \sum_t S_t Q_b S_t$, we have
\[
\frac{1}{d} \sum_{j,k \in [d]} U^b_j (U^b_k)\hc \otimes \ketbra{j}{k} 
    \quad = \quad \frac{1}{d} \sum_{j,k \in [d]}  \sum_{t \in [d']} 
        \bra{f_t} U^b_j (U^b_k)\hc \ket{f_t} \; \density{f_t} \otimes \ketbra{j}{k} \enspace.
\]
Thus
\[
U^b_j (U^b_k)\hc \quad = \quad \sum_{t \in [d']} 
        \bra{f_t} U^b_j (U^b_k)\hc \ket{f_t} \; \density{f_t}
\]
for every~$b,j,k \in [d]$. Taking~$k = 1$, and noting that~$U^b_1 = \id$ for all~$b$, we see that all the operators~$U^b_j$ are diagonal in the basis~$\set{\ket{f_t}}$. Thus, they all commute.
\end{proof}

We present examples of MUMs in the next section which violate the commutation relations in Eq.~\eqref{eq:U_comm}, and therefore cannot be expressed as a direct sum of MUBs.

\subsection{Examples of MUMs}
\label{sec-mum-examples}

In this section, we provide explicit MUM constructions for $d=4,5,6$ that are not direct sums of MUBs. The examples for~$d = 4,5$ were first reported in Ref.~\cite[Sec.~II.C.2, Supplementary materials]{TFR+21}, essentially in terms of the operators~$V^b_{jk}$. The example for~$d = 6$ is new.

We provide these constructions in terms of the unitary operators~$U^b_j$ satisfying the MUM conditions in Eq.~(\ref{eq:MUM_U}) as well as Eq.~(\ref{eq-canonical-U}), so that the MUMs are in canonical form. We leave it to the reader to verify that the conditions are satisfied. In all the constructions, the Hilbert space is $\bC^2 \otimes \bC^d$, and we use the single-qubit Pauli matrices $\id$, $X$, $Y$ and $Z$:
\[
	\id \eqdef \begin{pmatrix} 1 & 0 \\ 0 & 1 \end{pmatrix} \enspace, \qquad \PauliX \eqdef \begin{pmatrix} 0 & 1 \\ 1 & 0 \end{pmatrix} \enspace, \qquad  \PauliY \eqdef \begin{pmatrix} 0 & - \complexi \\ \complexi & 0 \end{pmatrix} \enspace, \text{ and} \qquad  \PauliZ \eqdef \begin{pmatrix} 1 & 0 \\ 0 & -1 \end{pmatrix} \;.
\]
We provide the unitary operators~$U^b_j$ in one block matrix $H_d$, whose $(b,j)$-block corresponds to~$U^b_j$, i.e., $H_d \eqdef \sum_{b,j \in [d]} \ketbra{b}{j} \otimes U^b_j$, where $( \ket{b} : b \in [d])$ is the standard basis of $\bC^d$.
The MUMs for~$d = 4$ are given by
\begin{equation}\label{eq:d4}
H_4 \quad \eqdef  \quad
\begin{pmatrix}
\I & \I & \I & \I \\
\I & \frac{2\ii}{3} (Z-Y) - \frac{1}{3} \I & \frac{2\ii}{3} (X-Z) - \frac{1}{3} \I & \frac{2\ii}{3} (Y-X) - \frac{1}{3} \I \\
\I & \ii Y & -\I & -\ii Y \\
\I & -\frac23 \I - \frac{\ii}{3} (2Z + Y) & \frac13 \I +\frac{2\ii}{3} (Z - X) & -\frac{2}{3} \I + \frac{\ii}{3}(2X + Y)
\end{pmatrix} \enspace.
\end{equation}
%
%
The MUMs for~$d = 5$ are given by
%
\begin{equation}\label{eq:d5}
H_5 \quad \eqdef  \quad
\begin{pmatrix}
\I & \I & \I & \I & \I \\
\I & -\I & -\ii cY -\ii sX & -\ii cY + \ii sX & -\ii Y \\
\I & -\ii cY - \ii sX & -\I & -\ii Y & -\ii cY + \ii sX \\
\I & -\ii cY + \ii sX & -\ii Y & -\I & -\ii cY - \ii sX \\
\I & -\ii Y & -\ii cY + \ii sX & -\ii cY - \ii sX & -\I 
\end{pmatrix} \enspace,
\end{equation}
where $s \eqdef \sin( \frac{2\pi}{3} )$ and $c \eqdef \cos( \frac{2\pi}{3} )$.
%
%
The new example of MUMs, for~$d = 6$, is given by
%
\begin{equation}\label{eq:d6}
H_6 \quad \eqdef  \quad
\begin{pmatrix}
\I & \I & \I & \I & \I & \I \\
\I & -\I & -\ii Z & \ii Z & \ii Y & -\ii Y \\
\I & -\ii X & -\I & \ii Y & -\ii Y & \ii X \\
\I & \ii X & \ii Y & -\I & -\ii Y & -\ii X \\
\I & \ii Y & -\ii Y & -\ii Y & -\I & \ii Y \\
\I & -\ii Y & \ii Z & -\ii Z & \ii Y & -\I 
\end{pmatrix} \enspace.
\end{equation}

By \Cref{prop-direct-sum-mub}, checking that these MUMs are not direct sums of MUBs reduces to checking that not all operators~$U^b_j$ commute. We may check this readily.

\Cref{lem-multiple-dim-mum} implies that there are~$d$-outcome MUMs which are not direct sums of MUBs whenever~$d$ is a multiple of~$4, 5, 6$.

\subsection{MUMs from Hadamard matrices}
\label{sec-mum-hadamard}

Here we describe a recipe for constructing MUMs from generalized Hadamard matrices, via representations of suitable associative algebras, in particular of the algebra of quaternions.

In general, MUMs correspond to a collection of unitary matrices satisfying Eq.~\eqref{eq:MUM_U}. If we remove the first condition, namely that~$U^b_1 = \id$ for all $b$, and add the orthogonality condition in Eq.~\eqref{eq:orthogonality_U}, we obtain the conditions in Eq.~\eqref{eq:unitary_Hadamard}. These conditions describe a collection of unitary matrices strongly resembling the orthogonality properties of Hadamard matrices. We call a block matrix with such unitary operators a \textit{Hadamard matrix of unitary operators}.
\begin{definition}
A \emph{Hadamard matrix of unitary operators\/} of size $d$ with block size $k$ is a~$(dk) \times (dk)$ block matrix~$\sum_{b,j = 1}^d \ketbra{b}{j} \tensor U^b_j$ whose blocks $U^b_j$, for~$b,j \in [d]$, are operators on a $k$-dimensional Hilbert space satisfying the relations \eqref{eq:unitary_Hadamard}.
\end{definition}
We recover complex Hadamard matrices by setting the block size~$k$ to $1$ (and real Hadamard matrices by only allowing $U^b_j = \pm1$). The examples presented in Eqs.~\eqref{eq:d4}, \eqref{eq:d5} and \eqref{eq:d6} are all Hadamard matrices of unitary operators with block size~$k=2$. Note that a similar definition has been considered by Banica~\cite{Ban18}, with the additional constraints
\begin{equation*}
\begin{split}
\sum_b (U^b_j)^\dagger U^b_k & \left. \quad = \quad \updelta_{jk} d\, \id \qquad \forall j,k \right. \\
\sum_j U^b_j (U^{b'}_j)^\dagger & \left. \quad = \quad \updelta_{bb'} d\, \id \qquad \forall b,b' \enspace. \right.
\end{split}
\end{equation*}
These conditions impose orthogonality-type conditions for block-rows (and block-columns) of unitary matrices with respect to both possible orders of multiplication of the matrices. 

We draw on a correspondence between quaternions and the single-qubit Pauli operators to give a general construction of Hadamard matrices of unitary operators with block size $k = 2$. Recall that the real algebra of quaternions may be represented as
\begin{equation*}
\bH \quad \eqdef \quad \{ \alpha + \beta \bi + \gamma \bj + \delta \bk ~|~  \alpha, \beta, \gamma, \delta \in \reals \} \enspace,
\end{equation*}
where~$\bi, \bj, \bk$ are the \emph{basic quaternions\/}, and satisfy the relations 
\[
\bi^2 = \bj^2 = -1, \quad \bi \bj = \bk, \quad \bj \bi = -\bk \enspace.
\]
(Note the distinction between the quaternion $\bi$ and the complex imaginary unit $\ii$.)
The algebra is endowed with a conjugation operation which satisfies~$\alpha^\ast = \alpha$ for~$\alpha \in \reals$, $\bi^\ast = -\bi$, $\bj^\ast = -\bj$ (and hence, $\bk^\ast = -\bk$). Conjugation also distributes over addition, and satisfies~$(\bq_1 \bq_2)^\ast = \bq_2^\ast \bq_1^\ast$. The norm~$\norm{\bq}$ of a quaternion~$\bq$ is given by
\[
\norm{\bq} \quad \eqdef \quad \sqrt{\bq \bq^\ast} \quad = \quad \sqrt{\bq^\ast \bq} \enspace, 
\]
and when~$\bq \neq 0$, its multiplicative inverse~$\bq^{-1}$ is given by~$\bq^{-1} \eqdef \bq^\ast / \norm{\bq}$. The set of quaternions with norm~$1$ are called \emph{unit quaternions\/}. The unit quaternions form a (non-commutative) group under the multiplication operation.

For a~$d \times d$ matrix~$M$ over~$\bH$, define~$M^\dagger$ as its conjugate-transpose, i.e., the matrix given by~$(M^\dagger)_{ij} \eqdef M_{ji}^\ast$ for all~$i,j \in [d]$. 
\begin{definition}
\label{def-q-hadamard}
A~$d \times d$ matrix quaternionic matrix~$M$ is called a \emph{quaternionic Hadamard matrix\/} if all its entries are unit quaternions, i.e., $ \norm{M_{ij}} = 1$ for all~$i,j \in [d]$, and~$M M^\dagger = d \, \id$.

We say that a quaternionic Hadamard matrix~$M$ is \emph{dephased\/} if the elements in its first row and the first column are all~$1$, i.e., $M_{1j} = M_{i1} = 1$ for all~$i,j \in [d]$.
\end{definition}
The condition~$ M M^\dagger = d \, \id$ is equivalent to~$ \sum_{j = 1}^d M_{bj} M_{b' j}^\ast = \updelta_{b b'} d \, \id$ for all~$b, b' \in [d]$. Note that the order of multiplication of matrices in this definition is important, since multiplication in~$\bH$ is not commutative.

Consider the linear map $f: \bH \to \mathcal{L}(\bC^2)$, from quaternions to linear operators on $\bC^2$, defined on the quaternion basis elements as
\begin{equation*}\label{eq:isometry}
f(1) \eqdef \I, \quad f( \bi ) \eqdef \ii X, \quad f( \bj ) \eqdef -\ii Y, \quad f( \bk ) \eqdef \ii Z \enspace.
\end{equation*}
We may verify that~$f$ is a group isomorphism between the quaternion group
\[
\sQ_8 \quad \eqdef \quad \{ \pm 1, \pm \bi, \pm \bj, \pm \bk \}
\]
of~$\bH$ and the subgroup $\sP \eqdef \{ \pm \I, \pm \ii X, \pm \ii Y, \pm \ii Z \}$ of the Pauli group. This isomorphism extends to a group isomorphism between the multiplicative group of unit quaternions and the subgroup of unitary operators in the real algebra generated by the subgroup~$\sP$. In particular, the function~$f$ commutes with conjugation, i.e., $f(\bq^\ast) = (f(\bq))^\dagger$ for all unit~$\bq \in \bH$. This further extends to an algebra isomorphism between~$\bH$ and the real algebra generated by~$\sP$. (The latter algebra is a strict subset of the space of~$2 \times 2$ complex matrices. For example, the Pauli operator~$X$ does not belong to this algebra.) Therefore, through the isomorphism $f$, every quaternionic Hadamard matrix can be mapped to a Hadamard matrix of unitary operators with block size 2. 
\begin{lemma}
\label{lem-qhadamard}
For any~$d \times d$ quaternionic Hadamard matrix~$M$, the~$(2d) \times (2d)$ block matrix~$H$ defined as
\[
H \quad \eqdef \quad \sum_{b,j = 1}^d \ketbra{b}{j} \tensor f( M_{bj}^\ast ) \enspace,
\]
is a Hadamard matrix of unitary operators of size~$d$ with block size~2.
\end{lemma}
\begin{proof}
Since~$M_{bj}$ is a unit quaternion for every~$b,j \in [d]$, $f(M_{bj}^\ast)$ is a unitary operator. Further, for~$b, b' \in [d]$, we have
\begin{equation}
\label{eq-qH-3}
\sum_{j = 1}^d f( M_{bj}^\ast )^\dagger \, f( M_{b' j}^\ast )
    \quad = \quad \sum_j f( M_{bj} ) \, f( M_{b' j}^\ast )
    \quad = \quad f \! \left( \sum_j  M_{bj} M_{b' j}^\ast \right)
    \quad = \quad \updelta_{b b'} d \, \id \enspace,
\end{equation}
as~$f$ is a unital algebra isomorphism. This is the third property in Eq.~\eqref{eq:unitary_Hadamard}. 
Consider the matrix~$K$ defined as
\[
K \quad \eqdef \quad \frac{1}{\sqrt{d}} \sum_{b,j = 1}^d \ketbra{j}{b} \tensor f( M_{bj}^\ast ) \enspace.
\]
The property in Eq.~\eqref{eq-qH-3} is equivalent to~$K^\dagger K = \id$, i.e., $K$ is unitary. So~$K K^\dagger = \id$, which is equivalent to the second property in Eq.~\eqref{eq:unitary_Hadamard}.
\end{proof}
Furthermore, the canonical form of a pair of MUMs corresponds to a \textit{dephased} quaternionic Hadamard matrix. According to \Cref{prop:dephasing}, this dephasing can be done by multiplying every row of the Hadamard matrix by the conjugate of its first element from the right, and every column by the conjugate of its first element from the left; the resulting matrix is also a Hadamard matrix. Due to the isomorphism~$f$ and \Cref{prop-direct-sum-mub}, the corresponding MUMs are a direct sum of MUBs if and only if all elements in the dephased quaternionic Hadamard matrix commute pairwise. Thus, we get the following correspondence.
\begin{theorem}
\label{thm-non-directsum}
Any~$d \times d$ \emph{dephased\/} quaternionic Hadamard matrix~$M$ with at least one pair of non-commuting elements defines a pair of~$d$-outcome MUMs that are not a direct sum of MUBs.
\end{theorem}

\subsection{Further examples of MUMs}
\label{sec-new-mums}

Using constructions of quaternionic Hadamard matrices in the literature, we exhibit infinitely many new MUMs. Following \Cref{thm-non-directsum}, we focus on Hadamard matrices which have non-commuting elements after dephasing.

Chterental and {\DJ}okovi{\'c}~\cite{CD08-stochastic-matrices} present two infinite families of~$4 \times 4$ quaternionic Hadamard matrices. The first is a \emph{special family\/} paramatrized by two unit quaternions~$\ba, \bb$, where
\begin{align*}
\ba \quad & \in \quad \set{ \alpha_1 + \alpha_2 \bi : \alpha_1, \alpha_2 \in \reals } \enspace, \qquad \text{and} \\
\bb \quad & \in \quad \set{ \beta_1 + \beta_2 \bj : \beta_1, \beta_2 \in \reals } \enspace.
\end{align*}
The Hadamard matrix corresponding to~$\ba, \bb$ is
\begin{align*}
M_{\ba, \bb} \quad \eqdef \quad
\begin{pmatrix}
1 & 1 & 1 & 1 \\
1 & -1 & \bb & -\bb \\
1 & \ba & \bx & \bz \\
1 & -\ba & \by & \bw
\end{pmatrix} \enspace,
\end{align*}
where
\begin{align*}
\bx \quad & \eqdef \quad - \frac{1}{2} (1 + \ba + \bb - \ba \bb) &
\bz \quad & \eqdef \quad - \frac{1}{2} (1 + \ba - \bb + \ba \bb) \\
\by \quad & \eqdef \quad - \frac{1}{2} (1 - \ba + \bb + \ba \bb) &
\bw \quad & \eqdef \quad - \frac{1}{2} (1 - \ba - \bb - \ba \bb) \enspace.
\end{align*}
Whenever~$\ba$ and~$\bb$ are both not real, they do not commute. The corresponding Hadamard matrix of unitary operators gives us~$4$-outcome MUMs which are not direct sums of MUBs. The second \emph{generic\/} family of~$4 \times 4$ quaternionic Hadamard matrices in Ref.~\cite{CD08-stochastic-matrices} similarly gives us MUMs of this type. (Chterental and {\DJ}okovi{\'c} follow a different convention in the definition of a quaternionic Hadamard matrix~$M$, viz., they require that~$ M^\dagger M = d \, \id$. We take the conjugate-transpose of their constructions to get Hadamard matrices as defined in this article.)

Barrera Acevedo and Dietrich~\cite[Lemma 5]{BAD18} point out a one-to-one correspondence between $d \times d$ circulant quaternionic Hadamard matrices and \emph{perfect sequences\/} of unit quaternions of length $d$.
\begin{definition}
A \emph{perfect sequence\/} of quaternions of length $d$ is a sequence~$Q \eqdef (\bq_0, \bq_1, \dotsc, \bq_{d-1})$ of quaternions such that its periodic $t$-autocorrelation values
\begin{equation}
\mathrm{AC}_Q(t) \quad \eqdef \quad \sum_{l=0}^{d-1} \bq_l \bq^\ast_{(l+t \mod d)}
\end{equation}
satisfy $\mathrm{AC}_Q(t) = 0$ for all $t \neq 0 \mod d$.
\end{definition}
A $d \times d$ circulant quaternionic Hadamard matrix~$M$ can be constructed from a perfect sequence~$(\bq_l)$ of unit quaternions of length $d$ by taking the first row of the matrix to be the sequence, and the rest of the rows to be successive cyclic shifts of the first row. That is, we define~$M_{ij} \eqdef \bq_{(i + j - 2 \mod d)}$ for~$i,j \in [d]$.

Perfect sequences have been found for small prime lengths~$5, 7, 9, 11, 13, 17, 19$, and~$23$ by Kuznetsov~\cite[Example 7.3]{Kuz17}. It turns out that all of these sequences correspond to MUMs that are not direct sums of MUBs. We first demonstrate this on the smallest example of length~$5$.
\begin{example}
\label{ex-ps5}
The perfect sequence of quaternions $(1, \bj, \bj, 1, \bq)$ with
\[
\bq \quad \eqdef \quad \frac{ -1 + \bi - \bj - \bk }{2}
\]
corresponds to a pair of 5-outcome MUMs that are not a direct sum of MUBs. To see this, note that~$\bq$ is a unit quaternion, and consider the corresponding circulant Hadamard matrix
\begin{equation}
\begin{pmatrix}
1 & \bj & \bj & 1 & \bq \\
\bq & 1 & \bj & \bj & 1 \\
1 & \bq & 1 & \bj & \bj \\
\bj & 1 & \bq & 1 & \bj \\
\bj & \bj & 1 & \bq & 1 \\
\end{pmatrix}.
\end{equation}
We dephase the first row, that is, multiply every column by the conjugate of its first element from the left to get
\begin{equation*}
\begin{pmatrix}
1 & 1 & 1 & 1 & 1 \\
\bq & -\bj & 1 & \bj & \bq^\ast \\
1 & - \bj \bq & -\bj & \bj & \bq^\ast \bj \\
\bj & -\bj & -\bj \bq & 1 & \bq^\ast \bj \\
\bj & 1 & -\bj & \bq & \bq^\ast \\
\end{pmatrix}.
\end{equation*}
Lastly, we dephase the first column by multiplying every row by the conjugate of the first element of that row from the right to get
\begin{equation*}
\begin{pmatrix}
1 & 1 & 1 & 1 & 1 \\
1 & -\bj \bq^\ast & \bq^\ast & \bj \bq^\ast & \bq^\ast \bq^\ast \\
1 & - \bj \bq & -\bj & \bj & \bq^\ast \bj \\
1 & -1 & \bj \bq \bj & -\bj & \bq^\ast \\
1 & -\bj & -1 & -\bq \bj & - \bq^\ast \bj \\
\end{pmatrix}.
\end{equation*}
These two transformations preserve the Hadamard property, according to \Cref{prop:dephasing}. 
The result is a quaternionic Hadamard matrix with non-commuting elements. For example, the last two elements of the second row, $\bj \bq^\ast$ and $\bq^\ast \bq^\ast$, do not commute. To see this, note that~$ \bj \bq^\ast$ and~$ \bq^\ast \bq^\ast$ commute iff~$\bj$ and~$\bq^\ast \bq^\ast$ commute. The property follows by noting that the square of a quaternion of the form~$(1/2)(\pm 1 \pm \bi \pm \bj \pm \bk)$ is also of the same form.

By \Cref{lem-qhadamard}, we obtain a Hadamard matrix of unitary operators that are in canonical form. By the properties of~$f$, some of these unitary operators do not non-commute. Hence, we get a pair of~$5$-outcome MUMs that are not a direct sum of MUBs. 
\end{example}

A similar argument shows that all the examples in Ref.~\cite[Example 7.3]{Kuz17} give rise to MUMs that are not direct sums of MUBs. In particular, all of the examples are perfect sequences of the form $(1, \bj, \ldots, \bj, 1, \bq)$ where 
\[
\bq \quad \in \quad \set{ \frac{ \pm 1 \pm \bi \pm \bj \pm \bk }{2} } \enspace.
\]
The elements in the sequence that we have left unspecified all belong to the quaternion group~$\sQ_8$. Quaternions~$\bq$ as above all have unit norm. The corresponding circulant matrix and its dephasing gives us
\begin{equation}\label{eq:q_Hadamard_gen}
\begin{pmatrix}
1 & \bj & \ldots & \bj & 1 & \bq \\
\bq & 1 & \bj & \ldots & \bj & 1 \\
\vdots & \vdots & \vdots & \vdots & \vdots & \vdots 
\end{pmatrix}
\to
\begin{pmatrix}
1 & 1 & \ldots & 1 & 1 & 1 \\
\bq & -\bj & \ldots & \ldots & \bj & \bq^\ast \\
\vdots & \vdots & \vdots & \vdots & \vdots & \vdots 
\end{pmatrix}
\to
\begin{pmatrix}
1 & 1 & \ldots & 1 & 1 & 1 \\
1 & -\bj \bq^\ast & \ldots & \ldots & \bj \bq^\ast & \bq^\ast \bq^\ast \\
\vdots & \vdots & \vdots & \vdots & \vdots & \vdots 
\end{pmatrix}.
\end{equation}
We have~$\bq^\ast \in \{ (\pm 1 \pm \bi \pm \bj \pm \bk)/2 \}$, and thus $\bj \bq^\ast$ and $\bq^\ast \bq^\ast$ do not commute, by the same reasoning as in \Cref{ex-ps5}. Therefore, the dephased quaternionic Hadamard matrix in Eq.~\eqref{eq:q_Hadamard_gen} defines a pair of MUMs that are not a direct sum of MUBs. From the examples in Ref.~\cite{Kuz17}, \Cref{thm-non-directsum}, and \Cref{lem-multiple-dim-mum} we thus get MUMs that are not direct sums of MUBs whenever the outcome number $d$ is a multiple of~$5, 7, 9, 11, 13, 17, 19$, or~$23$.

Bright, Kotsireas, and Ganesh~\cite{BKG20-perfect-sequences} construct perfect quaternion sequences of length~$2^t$ for all~$t \ge 0$. Example~13 in their article explicitly lists the sequences given by their construction for~$t \in [0,7]$. The sequences for~$t = 5,6,7$ all yield dephased Hadamard matrices with non-commuting elements in the second row. All the elements in these sequences belong to~$\sQ_8$. All three sequences have occurrences of each of the three basic quaternions~$\bi, \bj, \bk$ immediately followed by~$\pm 1$, have~$-1$ as the first element, and~$\bi$ or~$1$ as the last element. Thus, the second row of the corresponding dephased Hadamard matrices have occurrences of each of~$\bi, \bj, \bk$. We include the perfect sequence for~$t = 5$ below as an illustration:
\[
(-1,\bi,1,1,\bj,-\bj,-1,-\bk,-1,\bk,-1,\bj,\bj,-1,1,-\bi,-1,-\bi,1,-1,\bj,\bj,-1,\bk,-1,-\bk,-1,-\bj,\bj,1,1,\bi) \enspace.
\]
Their construction yields the following perfect sequence for~$t = 4$, which also gives us a dephased Hadamard matrix with non-commuting elements:
\[
(-1,1,\bi,-\bj,1,\bj,\bi,-1,-1,-1,\bi,\bj,1,-\bj,\bi,1) \enspace.
\]

The construction of perfect sequences of length~$2^t$ described in Theorem~7 of Ref.~\cite{BKG20-perfect-sequences} is recursive, and implies the following relationship between the sequence~$Q_t$ of length~$2^t$ and the sequence~$Q_{t+2}$ of length~$2^{t+2}$, for all~$t \ge 2$. The projection of~$Q_{t+2}$ to the even indices in~$[0, 2^{t+2} - 1]$ equals~$(Q_t, Q_t)$. I.e., if we delete the entries of~$Q_{t+2}$ given by the odd indices, we get~$Q_t$ concatenated with itself. This implies that deleting every second row and column from the upper-left submatrix of the dephased Hadamard matrix~$M_{t+2}$ obtained from~$Q_{t+2}$ gives us the dephased Hadamard matrix~$M_t$ obtained from~$Q_t$. I.e., $(M_t)_{ij} = (M_{t+2})_{2i-1, 2j-1}$ for~$i,j \in [2^t]$. The perfect sequence of length~$2^t$ starting from the sequence for~$t = 4$ or~$t = 5$ listed above thus yields a~$2^t \times 2^t$ dephased quaternionic Hadamard matrix with non-commuting elements for any~$t \ge 4$. These matrices all give us MUMs that are not direct sums of MUBs. The resulting MUMs are different from the ones obtained from~$H_4$ in \Cref{sec-mum-examples}. This provides further evidence for the prevalence of perfect sequences, quaternionic Hadamard matrices, and MUMs that are not direct sums of MUBs.

There are a number of other constructions of quaternionic Hadamard matrices in the literature; see, e.g., Ref.~\cite{BAD19}. It is likely that there be many more families of such MUMs.

\subsection{The number of MUMs}

Given the close connection of MUMs to MUBs, it is natural to ask what the maximal number of $d$-outcome measurements is such that they are pairwise mutually unbiased (i.e., pairwise MUMs). Note that for MUBs, this is a long-standing open problem for composite dimensions \cite{Zau99,Zau11}. The only known generic upper bound for MUBs is $d+1$ in dimension $d$, which is known to be saturated in prime power dimensions \cite{WF89}. In composite dimensions, however, the best known generic lower bound is $p^r+1$, where $p^r$ is the smallest prime power factor in the prime decomposition of the dimension. There is no composite dimension in which the exact number of MUBs is known.

We prove that the situation is drastically different for MUMs: in fact, there exist an unbounded number of MUMs for every outcome number $d$. We show this via a construction using Hilbert spaces of unbounded dimension. We first make use of a result from Ref.~\cite[Proposition B.1]{KSTB+19}. Recall that~$ \omega_d $ is a primitive~$d$th root of unity, and let~$Z_d \eqdef \sum_{j=1}^{d} \omega^j_d \ketbraq{j}$ and $X_d \eqdef \sum_{j=1}^{d} \ketbra{j+1}{j}$ be the generalized Pauli~$Z$ and~$X$ operators, respectively, in dimension $d$. (In this definition, we ``round''~$d+1$ down to~$1$.)

\begin{proposition}[Proposition B.1~\cite{KSTB+19}]\label{prop:XZ}
Let $A_1$ and $A_2$ be unitary operators acting on a finite-dimensional Hilbert space $\cH$, satisfying $A_1^d = A_2^d = \I$ for some integer $d \ge 2$. Suppose that $A_1$ and $A_2$ satisfy the commutation relation
\begin{equation}\label{eq:MUM_comm}
A_1 A_2 \quad = \quad \omega_d A_2 A_1 \enspace.
\end{equation}
Then, $\dim \cH = d \cdot d'$ for some integer $d' \ge 1$ and there exists a unitary operator $U : \cH \to \bC^d \otimes \bC^{d'}$ such that
\begin{equation*}
U A_1 U^\dagger = Z_d \otimes \I_{d'} \qquad \text{and} \qquad U A_2 U^\dagger = X_d \otimes \I_{d'} \enspace.
\end{equation*}
\end{proposition}

Now we are in the position to state the result on the number of MUMs.

\begin{theorem}\label{thm:no_MUMs}
For any $d,n \in \mathbb{N}$ such that $d,n \ge 2$ there exist $n$ MUMs with outcome number $d$.
\end{theorem}

\begin{proof}
We start with the canonical construction of a pair of MUBs in dimension $d$, via the eigenbases of $Z_d$ and $X_d$. That is, we write these operators in their spectral decomposition,
\begin{equation*}
    Z_d = \sum_{a=1}^{d} \; \omega_d^a \; P_a, \qquad X_d = \sum_{b=1}^{d} \overline{\omega}_d^b \; Q_b \enspace,
\end{equation*}
where~$P_a \eqdef \density{a}$ and~$Q_b \eqdef \density{\chi_b}$, and~$(\ket{\chi_b})$ is the Fourier basis for~$\complex^d$.
Since $\{P_a\}$ and $\{Q_b\}$ are MUBs, they are also MUMs. From Definition \ref{def:MUM}, if $\{P_a\}$ and $\{Q_b\}$ are MUMs then so are $\{P_a \otimes \I\}$ and $\{Q_b \otimes \I\}$, as are $\{ U P_a U^\dagger\}$ and $\{ U Q_b U^\dagger\}$ for an arbitrary unitary operator $U$. It follows from these observations and Proposition \ref{prop:XZ} that two unitary operators~$A_1$ and $A_2$ define a pair of $d$-outcome MUMs through their spectral decomposition if they satisfy $A^d_1 = A^d_2 = \I$ and the commutation relation in Eq.~\eqref{eq:MUM_comm}. Hence, in order to define~$n$ MUMs with~$d$ outcomes, it is sufficient to find $n$ unitary operators $\{A_j : j \in [n] \}$ such that $A^d_j = \I$ for all $j$ and $A_j A_k = \omega_d A_k A_j$ for all $j<k$.

For such a construction, consider the unitary operators $\{A_j\}$ defined on $( \bC^d )^{ \otimes n }$ as
\begin{equation*}
A_j \quad \eqdef \quad \left( \bigotimes_{k = 1}^{j - 1} X_d \right) \otimes Z_d \otimes \left( \bigotimes_{l = j + 1}^n \I_d \right) \enspace,
\end{equation*}
that is,
\begin{equation*}
\begin{split}
A_1 & \left. \quad = \quad Z_d \otimes \I_d \otimes \I_d \otimes \I_d \otimes \cdots \otimes \I_d \right. \\
A_2 & \left. \quad = \quad X_d \otimes Z_d \otimes \I_d \otimes \I_d \otimes \cdots \otimes \I_d \right. \\
A_3 & \left. \quad = \quad X_d \otimes X_d \otimes Z_d \otimes \I_d \otimes \cdots \otimes \I_d \right. \\
 & \left. \quad \vdots \right. \\
A_n & \left. \quad = \quad X_d \otimes X_d \otimes X_d \otimes X_d \otimes \cdots \otimes Z_d \enspace. \right.
\end{split}
\end{equation*}
It is straightforward to verify that these operators satisfy $A_j^d = \I$ for all $j$ and $A_j A_k = \omega_d A_k A_j$ for all $j<k$. Hence, they define $n$ MUMs with outcome number $d$.
\end{proof}

\section{Superdense coding and MUMs}
\label{sec-mum-protocol}

In this section, we describe some connections between MUMs and superdense coding protocols. We present the rigidity conjecture due to Nayak and Yuen in \Cref{sec-rigidity-conjecture}. We show how the encoding operators in optimal protocols can be constructed from any pair of MUMs in \Cref{sec-sdc-from-mum}. Then we prove the connection between rigidity of superdense coding and the expressibility of MUMs as a direct sum of MUBs in \Cref{sec-rigidity-mum}. This yields the counterexamples to the conjecture claimed in \Cref{sec-intro}.

\subsection{The rigidity conjecture}
\label{sec-rigidity-conjecture}

In a~$d$-dimensional superdense coding protocol, two parties, Alice and Bob, share a quantum state~$\tau$ on a bipartite Hilbert space~$\cH_A \otimes \cH_B$. We assume without loss of generality that~$\cH_{A}$ factors as~$\cH_{A'} \otimes \cH_{A''}$, where $\cH_{A''}$ is isomorphic to $\complex^d$. Given an input~$i \in [d^2]$, Alice applies a unitary operator $W_i$ (called an \emph{encoding operator\/}) to her share of~$\tau$ (with support in the space~$\cH_A$), and sends the qubits in~$A''$ to Bob. Bob then performs a measurement on the Hilbert space $\cH_{A''} \otimes \cH_{B}$ to determine what the input $i$ is. (Since Bob does not know the input~$i$, the measurement is independent of~$i$.) We may define the protocol formally as follows.
\begin{definition}[Superdense coding protocol]
\label{def:sdc-protocol}
Let $d$ be a positive integer.
Let $\cH_A \eqdef \cH_{A'} \otimes \cH_{A''}$ and~$\cH_B$ be finite dimensional Hilbert spaces where $\cH_{A''}$ is isomorphic to $\complex^d$. 
Let $\tau$ denote a quantum state on $\cH_A \otimes \cH_B$ and let~$(W_i : i \in [d^2])$ denote a sequence of $d^2$ unitary operators acting on $\cH_A$. We say that $(\tau,(W_i))$ is a \emph{$d$-dimensional superdense coding protocol\/} if there exists a measurement $( M_i : i \in [d^2])$ acting on $\cH_{A''} \otimes \cH_B$ such that~$\trace (M_i \, \rho_i) = 1$ for all~$ i \in [d^2]$, where $\rho_i$ denotes the reduced state on registers~$A'' B$, i.e., $\rho_i \eqdef \trace_{A'} [ (W_i \otimes \id) \tau (W_i \otimes \id)^\adjoint ]$. The operators~$(W_i)$ are then called the \emph{encoding operators\/}.
\end{definition}
A canonical protocol for $d$-dimensional superdense coding is as follows. Alice and Bob share the $d$-dimensional maximally entangled state $\ket{\mes_d} \eqdef \frac{1}{\sqrt{d}} \sum_{e=1}^{d} \ket{e} \ket{e}$. Given message~$i \in [d^2]$, Alice applies a unitary operator~$E_i$ to her share of $\ket{\mes_d}$, and sends it over to receiver. The family of unitary operators~$\{ E_i \}$ can be any orthogonal unitary basis for the space of~$d \times d$ complex matrices. (The orthogonality property means that $\trace(E_i^\adjoint E_j) = 0$ if and only if $i \neq j$.) An example of such a basis is the set of Heisenberg-Weyl operators (also called the generalized Pauli operators). The~$d^2$ possible states with which the receiver, Bob, ends up are then mutually orthogonal, and may be distinguished perfectly via a suitable measurement.

Werner~\cite{Werner01-teleportation-dense-coding} described, without proof, how non-equivalent orthogonal unitary bases may be constructed. Nayak and Yuen~\cite{NY20-rigidity} presented explicit constructions of such bases for dimensions three and higher, and proved their non-equivalence. In light of this, and the rigidity of the Bennett-Wiesner protocol that they established, they conjectured the rigidity of superdense coding for all dimensions, up to the choice of an orthogonal unitary basis. In more detail, they defined a notion of \emph{local equivalence\/}, and conjectured that every~$d$-dimensional superdense coding protocol is locally equivalent to one in which the sender measures a part of the shared entangled state to generate a shared random string~$r$, and then runs a protocol using an orthogonal unitary basis depending on~$r$.

More precisely, given an arbitrary protocol $(\tau,(W_i))$ for superdense coding, Nayak and Yuen conjectured that there exist local isometries $V, W$ such that if Alice applies~$V$ and Bob applies~$W$ to their respective shares of~$\tau$, then a maximally entangled state~$\ket{\mes_d}$ is extracted, in tensor product with an auxiliary state $\rho$. Alice then performs a projective measurement~$\{ P_r\}$ on her part of the auxiliary state~$\rho$ to obtain some outcome~$r$. Based on~$r$, Alice applies the~$i$th operator from a unitary orthogonal basis~$\set{E_{r,i}}$ to her half of the maximally entangled state. Finally, after sending her half of the state~$(E_{r,i} \tensor \id) \ket{\mes_d}$, the sender applies some unitary operator~$C_i$ on the remaining qubits in her possession. (This final operation does not affect Bob's reduced state or measurement in any way.) Bob measures his part of~$\rho$ with the same projective measurement~$\{ P_r\}$ as Alice. Using the outcome~$r$, he measures the remaining qubits appropriately to recover the classical message~$i$.

Formally, the rigidity conjecture may be stated as follows. In this statement, for any three operators~$C,D,E$, the notation~$C =_E D$ denotes that~$C E C^\adjoint = D E D^\adjoint$.
\begin{conjecture}[Nayak and Yuen~\cite{NY20-rigidity}]
\label{conj:high-dim-rigidity}
Let $(\tau,(W_i))$ be a $d$-dimensional superdense coding protocol. Then there exist 
\begin{enumerate}
	\item Unitary operators $V$ acting on $\cH_{A'} \otimes \cH_{A''}$ and $(C_i)_{i \in [d^2]}$ acting on $\cH_{A'}\,$,
	\item An isometry $W$ mapping $\cH_B$ to a Hilbert space $\cH_{B'} \otimes \cH_{B''}$ where $\cH_{B''}$ is isomorphic to $\complex^d$,
	\item A density matrix $\rho$ on $\cH_{A'} \otimes \cH_{B'}$,
	\item A set of pairwise orthogonal projectors $\{ P_r \}$ on~$\cH_{A'}\,$ that sum to the identity, and
	\item For every $r$, an orthogonal unitary basis $\{E_{r,i} \}_{i \in [d^2]}$ for the space of $d \times d$ complex matrices,
\end{enumerate}
such that, letting $\tau' \eqdef (V \otimes W) \tau (V \otimes W)^\adjoint$, we have
\[
	\tau' \quad = \quad \rho^{A'B'} \otimes \ketbra{\mes_d}{\mes_d}^{A'' B''}
\]
and for $i \in [d^2]$,
\[
	(C_i^\adjoint \otimes \id) W_i V^\adjoint \quad =_{\tau'} \quad \sum_{r} P_r \otimes E_{r,i}\;.
\]
\end{conjecture}
In other words, any $d$-dimensional superdense coding protocol is locally equivalent to a canonical protocol, in which the ancillary part of the shared state~$\tau$ in register~$A'$ serves as a source of randomness, and for each ``randomness string'' $r$, the encoding operators are given by an orthogonal unitary basis $\{E_{r,i}\}$.

\subsection{Encoding operators from MUMs}
\label{sec-sdc-from-mum}

Consider any pair of~$d$-outcome MUMs, $\{P_a\}_{a = 1}^d$ and $\{Q_b\}_{b = 1}^d$ in the form given by Eqs.~(\ref{eq:Pa}) and~(\ref{eq:Qb}).
By the characterization of MUMs described in \Cref{sec-mum-properties}, the space on which the MUMs act is isomorphic to~$\complex^n \tensor \complex^d$ for some positive integer~$n$. Define the maximally entangled state on $(\bC^n \tensor \complex^d)^{\tensor 2}$, corresponding to registers~$A$ and~$B$, as
\begin{equation*}
\ket{\mes^+} \quad \eqdef \quad \frac{1}{\sqrt{nd}} \sum_{p=1}^{nd} \ket{p}^A \ket{p}^B \enspace.
\end{equation*}
Define unitary operators
\begin{equation}
\label{eq-encoding-operators}
R \quad \eqdef \quad  \sum_{a=1}^d \omega^a P_a \qquad \text{and} \qquad S \quad \eqdef \quad  \sum_{b=1}^d \omega^b Q_b \enspace,
\end{equation}
where $\omega \eqdef \mathrm{e}^{\frac{2 \pi \ii}{d}}$ is a primitive~$d$-th root of unity. Consider the set $\{ W_{st} \}_{s,t = 1}^d$, where~$W_{st} \eqdef R^s S^t$. We claim that these give us a superdense coding protocol.
\begin{theorem}
\label{thm-mum-protocol}
The pair~$( \mes^+, (W_{st}))$ is a~$d$-dimensional superdense coding protocol.
\end{theorem}
\begin{proof}
We design a superdense coding protocol using the operators~$W_{st}$ as follows.
Registers~$A$ and~$B$ consist of subregisters~$A' A''$ and~$B' B''$, respectively. The subregisters~$A', B'$ are both~$n$-dimensional, and~$A'', B''$ are both~$d$-dimensional. 
In the superdense coding protocol, Alice and Bob hold registers~$A$ and~$B$, respectively, jointly prepared in state~$\ket{\mes^+}$. On input~$st \in [d] \times [d]$,
Alice applies the unitary operator~$W_{st}$ on her half of the maximally entangled state (in register~$A$), and sends the $d$-dimensional marginal state in subregister~$A''$ to Bob. 

Recall the property that for every linear operator $O$ on $\bC^{nd}$, we have $O \otimes \I_{nd} \ket{\mes^+} = \I_{nd} \otimes O\tran \ket{\mes^+}$, where the transposition is in the standard basis. Using this, we see that Bob ends up with the state~$\rho_{st}$ in registers~$A'' B$, where
\begin{align*}
\rho_{st} & \quad = \quad \trace_{A'} \big[ (W_{st} \otimes \I_{nd}) \ketbraq{\mes^+} (W^\adjoint_{st} \otimes \I_{nd}) \big] \\
    & \quad = \quad \trace_{A'} \big[ (\I_{nd} \otimes W_{st}\tran ) \ketbraq{\mes^+} (\I_{nd} \otimes \big(W^\adjoint_{st})\tran) \big] \\
    & \quad = \quad \frac{1}{n}\sum_{x=1}^n \ketbraq{\psi^x_{st}} \enspace,
\end{align*}
%
and the states~$\ket{\psi^x_{st}}$ are defined as
\begin{equation*}
\ket{\psi^x_{st}} \quad \eqdef \quad \frac{1}{\sqrt{d}} \sum_{l = 1}^d \ket{l} \otimes W_{st} \tran \ket{xl} \enspace. 
\end{equation*}
\suppress{
A straightforward computation shows that
\begin{equation*}
\ket{\psi^x_{st}} \quad = \quad \frac{1}{d \sqrt{d}} \sum_{b,j,k,l} \omega^{sj+tb} \ket{l} \otimes \Big[ \big( \overline{U^b_k} \big) \big(U^b_j \big) \tran \otimes \ketbra{k}{j} \Big] ( \ket{x} \otimes \ket{l} ) \enspace,  
\end{equation*}
%
}
Now
\begin{align}
\label{eq-inner-prod}
\braket{ \psi^x_{st} }{ \psi^{x'}_{s't'} } 
    & \quad = \quad \frac{1}{d} \bra{x} \left[ \trace_{B''} \left( \overline{W_{st}} W_{s' t'} \tran \right) \right] \ket{x'} \enspace,
\end{align}
where~$\overline{O}$ denotes the entry-wise complex conjugation of the operator~$O$ in the standard basis. Since~$\set{ P_a}$ and~$\set { Q_b}$ are a pair of projective measurements, the unitary operators $W_{st}$ can be written as
\suppress{
\begin{equation}\label{eq:Wst}
W_{st} \quad = \quad \frac1d \sum_{b,j,k = 1}^d \omega^{sj + tb} U^b_j (U^b_k)\hc \otimes \ketbra{j}{k} \enspace.
\end{equation}
}
\begin{equation}\label{eq:Wst}
W_{st} \quad = \quad \sum_{a,b = 1}^d \omega^{sa + tb} P_a Q_b \enspace.
\end{equation}
Since~$Q_b \tran = \overline{Q_b}$, and~$\set{ Q_b \tran}$ is also a projective measurement, we have
\begin{align*}
\trace_{B''} \left( \overline{W_{st}} W_{s' t'} \tran \right)
    & \quad = \quad \trace_{B''} \left( \sum_{a,b = 1}^d \omega^{- sa - tb} P_a \overline{Q_b} \sum_{a',b' = 1}^d \omega^{s'a' + t'b'} Q_{b'} \tran P_{a'}    \right) \\
    & \quad = \quad \sum_{a,b = 1}^d \omega^{ (s' - s)a + (t' - t)b} \trace_{B''} \left(  P_a Q_{b} \tran P_a \right) \enspace.
\end{align*}
Moreover~$P_a \tran = P_a$, and by definition of MUMs, we have~$P_a = P_a \tran = d P_a Q_b \tran P_a$.
Hence
\begin{align*}
\trace_{B''} \left( \overline{W_{st}} W_{s' t'} \tran \right)
    & \quad = \quad \frac{1}{d} \sum_{a,b = 1}^d \omega^{ (s' - s)a + (t' - t)b} \trace_{B''} \left( P_a \right) \\
    & \quad = \quad \updelta_{ss'} \updelta_{tt'} d \, \id \enspace,
\end{align*}
and by Eq.~(\ref{eq-inner-prod}), $\braket{ \psi^x_{st} }{ \psi^{x'}_{s't'} } = \updelta_{xx'} \updelta_{ss'} \updelta_{tt'}$.
This implies that
\begin{equation*}
\trace ( \rho_{st} \rho_{s't'} ) \quad = \quad \updelta_{s s'} \updelta_{t t'} \enspace,
\end{equation*}
i.e., the states~$\rho_{st}$ are mutually orthogonal. So there is a measurement that perfectly distinguishes these states, and~$( \mes^+, (W_{st}))$ is a~$d$-dimensional superdense coding protocol.
\end{proof}

\subsection{Rigidity and the direct-sum property}
\label{sec-rigidity-mum}

It turns out that \Cref{conj:high-dim-rigidity} implies that any superdense coding protocol derived from MUMs are direct sums of MUBs. We start by showing that the conjecture imposes a direct sum, i.e., block-diagonal structure on the encoding operators~$R,S$ defined in \Cref{sec-sdc-from-mum}.
\begin{proposition}
\label{prop-block-diagonal}
Suppose the MUMs~$\set{P_a}$ and~$\set{Q_b}$ are in canonical form.
If \Cref{conj:high-dim-rigidity} is true, then there is a basis~$(\ket{v_j} : j \in [n])$ for~$\complex^n$ and unitary operators~$(S_j : j \in [n])$ on~$\complex^d$ such that~$S = \sum_{j = 1}^n \density{v_j} \tensor S_j \,$.
\end{proposition}
\begin{proof}
Let~$T_0 \eqdef R$ and~$T_1 \eqdef S$. Note that~$R = \id \tensor Z_d$, where~$Z_d$ is the generalized Pauli~Z operator on~$\complex^d$ (also called the ``clock'' operator). Since the MUMs are in canonical form, we have~$Q_1 = \id \tensor \density{u}$, where~$\ket{u} \eqdef \tfrac{1}{\sqrt{d}} \sum_{j = 1}^d \ket{j}$.

Suppose \Cref{conj:high-dim-rigidity} holds. Then there are unitary operators~$V,W$ on~$\complex^n \tensor \complex^d$, unitary operators~$C_0, C_1$ on~$\complex^n$, an integer~$m \in [n]$, orthogonal projection operators~$(\Pi_r : r \in [m])$, and for each~$r \in [m]$ and orthogonal unitary basis~$(E_{ri} : i \in [d^2])$ for the vector space of linear operators on~$\complex^d$ such that for each~$i \in \set{0,1}$,
\begin{align*}
((C_i \tensor \id) T_i V^\adjoint) \tensor \id) (V \tensor W) \ket{\mes^+} 
    & \quad = \quad \left( \sum_{r = 1}^m \Pi_r \tensor E_{ri} \tensor \id) \right) (V \tensor W) \ket{\mes^+} \enspace.
\end{align*}
Since~$ (V \tensor W) \ket{\mes^+} = (\id \tensor W V \tran) \ket{\mes^+}$, we get
\begin{align*}
((C_i \tensor \id) T_i V^\adjoint) \tensor \id) \ket{\mes^+} 
    & \quad = \quad \left( \sum_{r = 1}^m \Pi_r \tensor E_{ri} \tensor \id) \right) \ket{\mes^+} \enspace.
\end{align*}
Since a unitary operator on~$\complex^n \tensor \complex^d$ is completely determined by its action on the maximally entangled state~$\ket{\mes^+}$, this is equivalent to
\begin{align*}
(C_i \tensor \id) T_i V^\adjoint 
    & \quad = \quad \sum_{r = 1}^m \Pi_r \tensor E_{ri} \enspace.
\end{align*}
Multiplying the adjoint of the operators for~$i = 0$ to those for~$i = 1$ on the right, we get
\begin{align*}
(C_1 \tensor \id) T_1 V^\adjoint V T_0^\adjoint (C_0^\adjoint \tensor \id) 
    & \quad = \quad \sum_{r = 1}^m \Pi_r \tensor E_{r1} E_{r0}^\adjoint \enspace.
\end{align*}
Using~$T_0 = R = \id \tensor Z_d$ and~$T_1 = S$, this is equivalent to
\begin{align*}
S & \quad = \quad \sum_{r = 1}^m (C_1^\adjoint \Pi_r C_0) \tensor ( E_{r1} E_{r0}^\adjoint Z_d) \enspace.
\end{align*}
By definition of~$S$ and the orthgonality of~$\set{ Q_b}$, we have~$Q_1 S Q_1 = \omega Q_1$. So
\begin{align*}
\sum_{r = 1}^m (C_1^\adjoint \Pi_r C_0) \bra{u} ( E_{r1} E_{r0}^\adjoint Z_d) \ket{u} \tensor \density{u} & \quad = \quad \omega \; \id \tensor \density{u} \enspace.
\end{align*}
Define~$\alpha_r \eqdef \bra{u} ( E_{r1} E_{r0}^\adjoint Z_d) \ket{u}$. The above equation implies that
\begin{align*}
C_1^\adjoint \left( \sum_{r = 1}^m \alpha_r \Pi_r \right) C_0 & \quad = \quad \omega \, \id \enspace, \\
\text{i.e.,} \qquad \sum_{r = 1}^m \alpha_r \Pi_r & \quad = \quad \omega C_1 C_0^\adjoint \enspace.
\end{align*}
In particular, the operator on the LHS above is unitary, and~$\abs{ \alpha_r} = 1$ for all~$r$.
Thus, $C_1^\adjoint = \omega C_0^\adjoint \sum_{r} \bar{\alpha_r} \Pi_r $, and
\begin{align*}
S & \quad = \quad \sum_{r} (C_0^\adjoint \Pi_r C_0) \tensor \left( \omega \bar{\alpha_r} E_{r1} E_{r0}^\adjoint Z_d \right) \enspace.
\end{align*}
The proposition follows.
\end{proof}

Finally, we show that the block-diagonal structure of the encoding operators carries over to the MUMs.
\begin{theorem}
\label{thm-rigidity-direct-sum}
Suppose the MUMs~$\set{P_a}$ and~$\set{Q_b}$ are in canonical form.
If \Cref{conj:high-dim-rigidity} is true, then the MUMs~$\set{ P_a}$ and~$\set{ Q_b}$ are a direct sum of MUBs.
\end{theorem}
\begin{proof}
By \Cref{prop-block-diagonal}, $S = \sum_{k = 1}^n \density{v_k} \tensor S_k$ for some orthonormal basis~$( \ket{v_k} )$ for~$\complex^n$, and unitary operators~$(S_k)$ on~$\complex^d$. Since~$S^d = \id$, we have~$S_k^d = \id$ for all~$k$. So all the eigenvalues of~$S_k$ are~$d$th roots of unity.

We may write~$P_a = \sum_k \density{v_k} \tensor \density{a}$. The operator~$Q_b$ inherits the same structure, as we may extract it from the encoding operator~$S$ as follows:
\begin{align*}
\label{eq-proj}
Q_b & \quad = \quad \frac{1}{d} \sum_{j = 1}^d \omega^{-bj} S^j \\
    & \quad = \quad \sum_{k = 1}^n \density{v_k} \tensor \frac{1}{d} \sum_{j = 1}^d \omega^{-bj} S_k^j \enspace.
\end{align*}
The operator~$Q_{bj} \eqdef \tfrac{1}{d} \sum_{j = 1}^d \omega^{-bj} S_k^j$ is the orthogonal projection onto the~$\omega^{b}$ eigenspace of~$S_k$. So~$Q_b$ is also a direct sum of orthogonal projection operators. The MUM property implies that for each~$j$, $\set{\density{a} : a \in [d]}$ and~$\set{ Q_{bj} : b \in [d]}$ are MUBs. The lemma follows.
\end{proof}

In \Cref{sec-mum-examples,sec-new-mums} we presented MUMs that are not direct sums of MUBs. It follows from \Cref{thm-rigidity-direct-sum} that \Cref{conj:high-dim-rigidity} is false for an infinite number of dimensions~$d \ge 4$. 

\bibliography{refs}

\newcommand{\etalchar}[1]{$^{#1}$}
\begin{thebibliography}{PWTR18}

\bibitem[BAD18]{BAD18}
Santiago Barrera~Acevedo and Heiko Dietrich.
\newblock Perfect sequences over the quaternions and $(4n, 2, 4n, 2n)$-relative
  difference sets in~${C}_n \times {Q}_8$.
\newblock {\em Cryptography and Communications}, 10(2):357--368, 2018.

\bibitem[BAD19]{BAD19}
Santiago Barrera~Acevedo and Heiko Dietrich.
\newblock New infinite families of {Williamson Hadamard} matrices.
\newblock {\em Australasian Journal of Combinatorics}, 73(1):207--219, January
  2019.

\bibitem[Ban18]{Ban18}
Teodor Banica.
\newblock Complex {Hadamard} matrices with noncommutative entries.
\newblock {\em Annals of Functional Analysis}, 9(3):354--368, 2018.

\bibitem[BKG20]{BKG20-perfect-sequences}
Curtis Bright, Ilias Kotsireas, and Vijay Ganesh.
\newblock New infinite families of perfect quaternion sequences and
  {Williamson} sequences.
\newblock {\em IEEE Transactions on Information Theory}, 66(12):7739--7751,
  December 2020.

\bibitem[BW92]{BW92-sdc}
Charles~H. Bennett and Stephen~J. Wiesner.
\newblock Communication via one- and two-particle operators on
  einstein-podolsky-rosen states.
\newblock {\em Physical review letters}, 69(20):2881, 1992.

\bibitem[C{\DJ}08]{CD08-stochastic-matrices}
Oleg Chterental and Dragomir~{\v Z}. {\DJ}okovi{\'c}.
\newblock On orthostochastic, unistochastic and qustochastic matrices.
\newblock {\em Linear Algebra and its Applications}, 428(4):1178--1201, 2008.

\bibitem[DEB{\.{Z}}10]{DEBZ10}
Thomas Durt, Berthold-Georg Englert, Ingemar Bengtsson, and Karol
  {\.{Z}}yczkowski.
\newblock On mutually unbiased bases.
\newblock {\em International Journal of Quantum Information}, 08(04):535--640,
  2010.

\bibitem[Gav08]{Gavinsky08-entanglement}
Dmitry Gavinsky.
\newblock On the role of shared entanglement.
\newblock {\em Quantum Information and Computation}, 8(1{\&}2):82--95, 2008.

\bibitem[KG14]{KG14}
Amir Kalev and Gilad Gour.
\newblock Mutually unbiased measurements in finite dimensions.
\newblock {\em New Journal of Physics}, 16(5):053038, 2014.

\bibitem[K{\v{S}}T{\etalchar{+}}19]{KSTB+19}
J{\k{e}}drzej Kaniewski, Ivan {\v{S}}upi{\'{c}}, Jordi Tura, Flavio Baccari,
  Alexia Salavrakos, and Remigiusz Augusiak.
\newblock Maximal nonlocality from maximal entanglement and mutually unbiased
  bases, and self-testing of two-qutrit quantum systems.
\newblock {\em Quantum}, 3:198, 2019.

\bibitem[Kuz10]{Kuz17}
Oleg Kuznetsov.
\newblock {\em Perfect Sequences over the Real Quaternions}.
\newblock PhD thesis, School of Mathematical Sciences, Monash University,
  Australia, 2010.

\bibitem[NPA12]{NPA12-sdp-relaxations}
Miguel Navascu{\'e}s, Stefano Pironio, and Antonio Ac{\'i}n.
\newblock {SDP} relaxations for non-commutative polynomial optimization.
\newblock In Miguel~F. Anjos and Jean~B. Lasserre, editors, {\em Handbook on
  Semidefinite, Conic and Polynomial Optimization}, chapter~21, pages 601--634.
  Springer US, Boston, MA, 2012.

\bibitem[NY20]{NY20-rigidity}
Ashwin Nayak and Henry Yuen.
\newblock Rigidity of superdense coding.
\newblock Technical Report arXiv:2012.01672v1 [quant-ph], arXiv Pre-print
  server, \texttt{https://arxiv.org/abs/2012.01672}, December 2020.

\bibitem[PWTR18]{PWTR18}
E.~C. Paul, S.~P. Walborn, D.~S. Tasca, and \L{}ukasz Rudnicki.
\newblock Mutually unbiased coarse-grained measurements of two or more
  phase-space variables.
\newblock {\em Physical Review A}, 97:052103, 2018.

\bibitem[RW21]{RW21}
\L{}ukasz Rudnicki and Stephen~P. Walborn.
\newblock Entropic uncertainty relations for mutually unbiased periodic
  coarse-grained observables resembling their discrete counterparts.
\newblock {\em Physical Review A}, 104:042210, 2021.

\bibitem[SRTW22]{SRTW22}
Thais~L. Silva, \L{}ukasz Rudnicki, Daniel~S. Tasca, and Stephen~P. Walborn.
\newblock Discretized continuous quantum-mechanical observables that are
  neither continuous nor discrete.
\newblock {\em Physical Review Research}, 4:013060, 2022.

\bibitem[TFR{\etalchar{+}}21]{TFR+21}
Armin Tavakoli, M{\'a}t{\'e} Farkas, Denis Rosset, Jean-Daniel Bancal, and
  J\k{e}drzej Kaniewski.
\newblock Mutually unbiased bases and symmetric informationally complete
  measurements in {Bell} experiments.
\newblock {\em Science Advances}, 7(7), 2021.

\bibitem[TSWR18]{TSWR18}
Daniel~S. Tasca, Piero S\'anchez, Stephen~P. Walborn, and \L{}ukasz Rudnicki.
\newblock Mutual unbiasedness in coarse-grained continuous variables.
\newblock {\em Physical Review Letters}, 120:040403, 2018.

\bibitem[Wer01]{Werner01-teleportation-dense-coding}
Reinhard~F. Werner.
\newblock All teleportation and dense coding schemes.
\newblock {\em Journal of Physics A: Mathematical and General},
  34(35):7081--7094, August 2001.

\bibitem[WF89]{WF89}
William~K. Wootters and Brian~D. Fields.
\newblock Optimal state-determination by mutually unbiased measurements.
\newblock {\em Annals of Physics}, 191(2):363--381, 1989.

\bibitem[Wil13]{W13-qit}
Mark~M. Wilde.
\newblock {\em Quantum Information Theory}.
\newblock Cambridge University Press, Cambridge, UK, 2013.

\bibitem[Zau99]{Zau99}
Gerhard Zauner.
\newblock {\em Quantendesigns: Grundz{\"u}ge einer nichtkommutativen
  Designtheorie}.
\newblock PhD thesis, University of Vienna, Austria, 1999.

\bibitem[Zau11]{Zau11}
Gerhard Zauner.
\newblock Quantum designs: Foundations of a non-commutative {Design Theory}.
\newblock {\em International Journal of Quantum Information}, 09(01):445--507,
  2011.
\newblock English translation of \cite{Zau99}.

\end{thebibliography}

\end{document}